\documentclass[11pt]{article}
\usepackage[margin=1in]{geometry}
\usepackage{fullpage}
\usepackage{times}
\usepackage[T1]{fontenc}
\usepackage{graphicx}
\usepackage{amsfonts,amsmath,amsthm,amssymb,dsfont,mathtools}
\usepackage{xfrac,nicefrac}
\usepackage{mathdots}
\usepackage{bm}
\usepackage{url}
\usepackage{paralist}
\usepackage[shortlabels]{enumitem}
\usepackage[normalem]{ulem}
\usepackage{xspace}
\xspaceaddexceptions{]\}}
\usepackage{comment}
\usepackage{tabu}
\usepackage{framed}
\usepackage{float,wrapfig}
\usepackage{subcaption}
\usepackage{tikz}
\usetikzlibrary{positioning}
\usepackage{algpseudocode} 
\usepackage[usenames,dvipsnames]{pstricks}
\usepackage[linesnumbered,boxed,ruled,vlined]{algorithm2e}
\usepackage{thm-restate}

\usepackage[colorlinks,citecolor=blue,linkcolor=blue,urlcolor=red,pagebackref]{hyperref}

\usepackage[capitalise]{cleveref}

\theoremstyle{plain}

\newtheorem{theorem}{Theorem}[section]
\newtheorem{lemma}[theorem]{Lemma}

\theoremstyle{definition}

\newtheorem{remark}[theorem]{Remark}
\newtheorem{property}[theorem]{Property}
\newtheorem{hypothesis}{Hypothesis}

\renewcommand{\epsilon}{\varepsilon}
\newcommand{\eps}{\varepsilon}

\newcommand{\tO}{\widetilde{O}}
\newcommand{\OO}{\widetilde{O}}
\newcommand{\ExactTri}{\textsc{Exact-Tri}}

\newcommand{\ThreeSUM}{\textsc{3sum}}
\newcommand{\ThreeSUMList}{\textsc{3sum-List}}

\newcommand{\SparseTriList}{\textsc{Sparse-Tri-List}}
\newcommand{\SparseTri}{\textsc{Sparse-Tri}}
\newcommand{\AESparseTri}{\textsc{AE-Sparse-Tri}}
\newcommand{\ANSparseTri}{\textsc{AN-Sparse-Tri}}
\newcommand{\ExactTriList}{\textsc{Exact-Tri-List}}

\title{Simpler Reductions from Exact Triangle}
\author{Timothy M. Chan\thanks{Supported by NSF Grant CCF-2224271.}\\UIUC\\tmc@illinois.edu \and Yinzhan Xu\thanks{Partially supported by NSF Grants CCF-2129139 and CCF-2330048 and BSF Grant 2020356.}\\MIT\\xyzhan@mit.edu}
\date{}

\begin{document}

\maketitle

\begin{abstract}
    In this paper, we provide simpler reductions from Exact Triangle to two important problems in fine-grained complexity: Exact Triangle with Few Zero-Weight $4$-Cycles and All-Edges Sparse Triangle. 

    Exact Triangle instances with few zero-weight $4$-cycles  was considered by Jin and Xu [STOC 2023], who used it as an intermediate problem to show $3$SUM hardness of All-Edges Sparse Triangle with few $4$-cycles (independently obtained by Abboud, Bringmann and Fischer [STOC 2023]), which is further used to show $3$SUM hardness of a variety of problems, including $4$-Cycle Enumeration, Offline Approximate Distance Oracle, Dynamic Approximate Shortest Paths and All-Nodes Shortest Cycles. We provide a simple reduction from Exact Triangle to Exact Triangle with few zero-weight $4$-cycles. Our new reduction not only simplifies Jin and Xu's previous reduction, but also strengthens the conditional lower bounds from being under the $3$SUM hypothesis to the even more believable Exact Triangle hypothesis. As a result, all conditional lower bounds shown by Jin and Xu [STOC 2023] and by Abboud, Bringmann and Fischer [STOC 2023] using All-Edges Sparse Triangle with few $4$-cycles as an intermediate problem now also hold under the Exact Triangle hypothesis. 

    We also provide two alternative proofs of the conditional lower bound of the All-Edges Sparse Triangle problem under the Exact Triangle hypothesis, which was originally proved by Vassilevska Williams and Xu [FOCS 2020]. Both of our new reductions are simpler, and one of them is also deterministic---all previous reductions from Exact Triangle or 3SUM to All-Edges Sparse Triangle (including P\u{a}tra\c{s}cu's seminal work [STOC 2010]) were randomized.
\end{abstract}

\section{Introduction}
In this paper, we present several new simpler reductions between core problems in fine-grained complexity, which lead to simpler conditional lower bound proofs for a plethora of problems, and at the same time strengthening some of these results by weakening the hypothesis assumed.

In the Exact Triangle problem (also known as Zero-Weight Triangle), we are given an $n$-node weighted graph $G$ whose edge weights are from $[\pm n^{O(1)}]$,\footnote{For a nonnegative integer $N$, $[N]$ denotes $\{1, \ldots, N\}$ and $[\pm N]$ denotes $\{-N, \ldots, N\}$.} and we need to determine whether it contains a triangle whose edge weights sum up to $0$. It is one of the key problems in the field of fine-grained complexity, primarily due to its relationship between the $3$SUM hypothesis and the APSP hypothesis, which are among the three central hypotheses in fine-grained complexity (the other one is the Strong Exponential Time hypothesis). It is known that under either the $3$SUM hypothesis \cite{VWfindingcountingj} or the APSP hypothesis \cite{focsyj}, the Exact Triangle problem requires $n^{3-o(1)}$ time. As a result, the following Exact Triangle hypothesis holds as long as at least one of the $3$SUM hypothesis and the APSP hypothesis holds.
\begin{hypothesis}[Exact Triangle hypothesis]
\label{hypo:exacttri}
In the Word-RAM model with $O(\log n)$-bit words, Exact Triangle on $n$-node weighted graphs whose edge weights are from $[\pm n^{O(1)}]$ requires $n^{3-o(1)}$ time. 
\end{hypothesis}

By designing reductions from the Exact Triangle problem to other problems, we can obtain conditional lower bounds that hold under the Exact Triangle hypothesis, which in turn also hold under either the $3$SUM hypothesis or the APSP hypothesis. This makes it desirable to design reductions from the Exact Triangle problem.

\paragraph{Exact Triangle with few zero-weight $4$-cycles.} Recently, Jin and Xu \cite{JinXstoc23} reduced $3$SUM to $4$-Cycle Enumeration, Offline Approximate Distance Oracle, Dynamic Approximate Shortest Paths and All-Nodes Shortest Cycles. As one of their intermediate steps, they showed $3$SUM-hardness of a variant of Exact Triangle on instances with a small number of zero-weight $4$-cycles. This part of their reduction is very technically involved, and it uses heavy machineries such as the Balog--Szemer{\'e}di--Gowers theorem from additive combinatorics. 
Independently, Abboud, Bringmann and Fischer \cite{AbboudBF23} also showed $3$SUM hardness of several graph problems including $4$-Cycle Enumeration and Offline Approximate Distance Oracle. They do not use aforementioned variant of Exact Triangle in their reduction, but their reduction also uses tools from additive combinatorics. 
(See Figure~\ref{fig:reductions} (right).)

\begin{figure}[ht]
\centering
\begin{subfigure}[b]{0.49\textwidth}
    \centering
    \scalebox{0.7}{
    \begin{tikzpicture}
		\node at(0, 0)  [anchor=center, align=center] (exacttri){Exact Triangle};
		\node at(-3, -1)  [anchor=center, align=center] (3sum){$3$SUM};
		\node at(3, -1)  [anchor=center, align=center] (apsp){APSP};	
		\node at(0, 4)  [anchor=center, align=center] (aesparsetri){All-Edges Sparse Triangle};	
		\node at(0, 5.5)  [anchor=center, align=center] (probs){many data structure problems \\ (dynamic reachability, set disjointness queries, etc.)};	
		
		\draw[->,line width=1pt] (3sum) to[]  node[below] {} (exacttri);
		\draw[->,line width=1pt] (apsp) to[]  node[below] {} (exacttri);
		\draw[->,line width=1pt] (exacttri) to[]  node(vx20)[right] {\cite{williamsxumono}} (aesparsetri);
		\node[below = of vx20.west, text = red, anchor = west, align = left] () {this paper \\ (simpler \& det.)};
		
		\draw[->,line width=1pt, text width = 1.6cm] (3sum) to[]  node()[left] {\cite{patrascu2010towards} or \cite{KopelowitzPP16}} (aesparsetri);
		
		\draw[->,dashed, bend right, line width=1pt] (apsp) to[]  node()[right] {or \cite{CVXstoc22}} (aesparsetri);

		\draw[->,line width=1pt] (aesparsetri) to[]  node[below] {} (probs.200);
		\draw[->,line width=1pt] (aesparsetri) to[]  node[below] {} (probs.230);
		\draw[->,line width=1pt] (aesparsetri) to[]  node[below] {} (probs.310);
		\draw[->,line width=1pt] (aesparsetri) to[]  node[below] {} (probs.340);
	\end{tikzpicture}}
\end{subfigure}
\hfill
\begin{subfigure}[b]{0.49\textwidth}
    \centering
    \scalebox{0.7}{
    \begin{tikzpicture}
		\node at(0, 0)  [anchor=center, align=center] (exacttri){Exact Triangle};
		\node at(-3, -1)  [anchor=center, align=center] (3sum){$3$SUM};
		\node at(3, -1)  [anchor=center, align=center] (apsp){APSP};	
		
		\node at(-2, 2)  [anchor=center, align=center] (3sumenergy){$3$SUM with\\ moderate energy};	

		\node at(-1, 4)  [anchor=center, align=center] (exacttriC4){Exact Triangle\\ with few zero-wt 4-cycles};	
		\node at(0, 6)  [anchor=center, align=center] (aesparsetriC4){All-Edges Sparse Triangle \\ with few 4-cycles};	
		\node at(0, 7.5)  [anchor=center, align=center] (probs){many more problems\\
($4$-cycle enumeration, approximate distance oracles, etc.)};	

		\draw[->,line width=1pt] (3sum) to[]  node[below] {} (exacttri);
		\draw[->,line width=1pt] (apsp) to[]  node[below] {} (exacttri);
		\draw[->,line width=1pt, text width = 3.1cm] (3sum) to[]  node()[left] {\cite{JinXstoc23} or \cite{AbboudBF23} (complicated!)} (3sumenergy);
		\draw[->,line width=1pt] (3sumenergy) to[]  node()[left] {\cite{JinXstoc23}} (exacttriC4);
		\draw[->,line width=1pt] (exacttriC4) to[]  node()[left] {\cite{JinXstoc23}} (aesparsetriC4);
		\draw[->,line width=1pt, bend left = 50] (3sumenergy.150) to[]  node()[left] {or \cite{AbboudBF23}} (aesparsetriC4);
		\draw[->,line width=1pt, text = red, text width = 3cm] (exacttri) to[]  node()[right] {this paper \quad \quad (new \& simpler!)} (0, 3.5);
		
		\draw[->,dashed, line width=1pt, text width = 2.6cm, bend right = 70, align = center] (exacttri.east) to[]  node()[right] {\cite{AbboudBKZ22} \quad \quad (but suboptimal)} (aesparsetriC4);
		
		\draw[->,line width=1pt] (aesparsetriC4) to[]  node[below] {} (probs.200);
		\draw[->,line width=1pt] (aesparsetriC4) to[]  node[below] {} (probs.230);
		\draw[->,line width=1pt] (aesparsetriC4) to[]  node[below] {} (probs.310);
		\draw[->,line width=1pt] (aesparsetriC4) to[]  node[below] {} (probs.340);
	\end{tikzpicture}}
\end{subfigure}
\caption{Summary of reductions.  (See the cited references for precise definitions of some of the intermediate problems in the right diagram, which may have slight variations in different papers.)}\label{fig:reductions}
\end{figure}
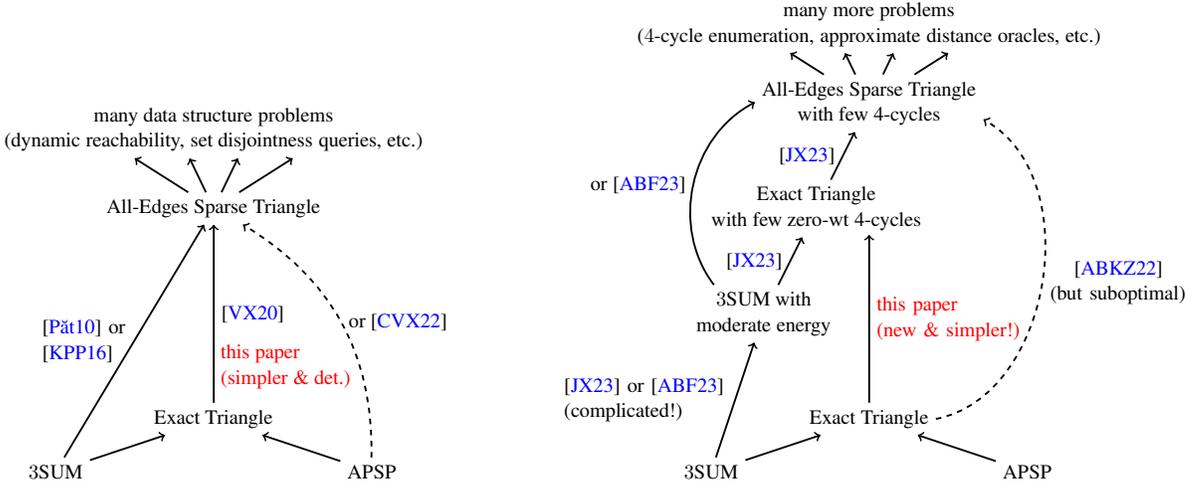

In this paper, we \emph{vastly} simplify Jin and Xu's reduction by designing a reduction from Exact Triangle (instead of $3$SUM in  their case) to the aforementioned variant of Exact Triangle, by a proof that is under two pages. Moreover, with our new reduction, the lower bounds for $4$-Cycle Enumeration, Offline Approximate Distance Oracle, Dynamic Approximate Shortest Paths and All-Nodes Shortest Cycles in \cite{JinXstoc23} now also hold under the Exact Triangle hypothesis; previously, they only hold under the $3$SUM hypothesis. (We remark that our new reduction does not imply Exact Triangle hardness for problems such as Sidon Set Verification and 4-LDT, which were shown to be $3$SUM hard in \cite{JinXstoc23}, because the reduction from $3$SUM to them does not use the variant of Exact Triangle as an intermediate problem.)

Our new reduction is inspired by the recent work of Chan, Vassilevska Williams and Xu \cite{CVXstoc23}. In particular, they showed versions of the Balog--Szemer{\'e}di--Gowers theorem that have simpler proofs inspired by the famous ``Fredman's trick''\footnote{The simple observation that $a+b=a'+b'$ is equivalent to $a-a'=b'-b$.} \cite{fredman1976new}.  We borrow their intuition that Fredman's trick can help simplifying proofs involving the Balog--Szemer{\'e}di--Gowers theorem, but our reduction does not use their results directly.

\paragraph{All-Edges Sparse Triangle.} We also present two simplifications of the known reduction from Exact Triangle to the All-Edges Sparse Triangle problem \cite{williamsxumono},  in which we are given an $m$-edge graph, and for each edge in the graph, we need to decide whether it is in a triangle.  This problem can be solved in $O(m^{2\omega / (\omega + 1)})$ time~\cite{AlonYZ97}, where $\omega < 2.372$ is the square matrix multiplication exponent \cite{DuanWZ23}. When $\omega = 2$, this running time becomes $O(m^{4/3})$. In a landmark paper, P{\u{a}}tra{\c{s}}cu \cite{patrascu2010towards} showed that this problem requires $m^{4/3-o(1)}$ time under the $3$SUM hypothesis, matching the running time if $\omega = 2$.\footnote{Technically, P{\u{a}}tra{\c{s}}cu proved the conditional lower bound for the problem of listing $m$ triangles in an $m$-edge graph, but this problem is equivalent up to polylog factors to the All-Edges Sparse Triangle problem under randomized reductions \cite{DurajK0W20}.} Vassilevska Williams and Xu \cite{williamsxumono} later strengthened this lower bound to make it also work under the Exact Triangle hypothesis.
Another paper by Chan, Vassilevska Williams and Xu \cite{CVXstoc22} showed that conditional lower bounds for All-Edges Sparse Triangle also hold under the real-valued variants of fine-grained hypotheses, including an $m^{4/3-o(1)}$ time lower bound under the Real APSP hypothesis, an $m^{5/4-o(1)}$ time lower bound under the Real Exact Triangle hypothesis, and an $m^{6/5-o(1)}$ time lower bound under the Real 3SUM hypothesis. Since P{\u{a}}tra{\c{s}}cu's work \cite{patrascu2010towards},  All-Edges Sparse Triangle has been used as a bridge to show conditional lower bounds for a wide variety of dynamic problems, such as  dynamic reachability and dynamic shortest paths \cite{patrascu2010towards, abboud2014popular}, and also static data structure problems, such as set disjointness and set intersection~\cite{KopelowitzPP16} and certain types of generalized range queries~\cite{DurajK0W20}.
(See Figure~\ref{fig:reductions} (left).) 

We design two  simple new reductions from Exact Triangle to All-Edges Sparse Triangle, simplifying Vassilevska Williams and Xu's reduction~\cite{williamsxumono}. In particular, our reductions do not use any hash function and are considerably shorter than theirs. 

Our first reduction is adapted from Chan, Vassilevska Williams and Xu's reduction from real-valued problems to All-Edges Sparse Triangle \cite{CVXstoc22}. Even though they were only able to obtain a non-tight $m^{5/4-o(1)}$ time lower bound for All-Edges Sparse Triangle under the Real Exact Triangle hypothesis, we show that their ideas can be implemented more efficiently for the integer-valued version of Exact Triangle. Similar to \cite{CVXstoc22}, our new reduction also relies on Fredman's trick. An additional advantage of this reduction is that it is deterministic, while the reduction in \cite{williamsxumono} is randomized. 
Note that P{\u{a}}tra{\c{s}}cu's original reduction~\cite{patrascu2010towards} from 3SUM to All-Edges Sparse Triangle
and Kopelowitz, Pettie and Porat's later reduction~\cite{KopelowitzPP16} were
also randomized (and crucially relied on hashing).  Our approach  gives the first deterministic reduction from 3SUM to All-Edges Sparse Triangle as well, yielding the same $m^{4/3-o(1)}$ conditional lower bound.

Our second reduction borrows an idea from a recent conditional lower bound of Sparse Triangle Detection under the strong $3$SUM hypothesis by Jin and Xu \cite{JinXstoc23}. Sparse Triangle Detection is a problem closely-related to All-Edges Sparse Triangle, where now we only need to decide if a given graph contains a triangle or not. Jin and Xu \cite{JinXstoc23} showed an $m^{6/5-o(1)}$ lower bound for Sparse Triangle Detection under the  Strong $3$SUM hypothesis, which states that $3$SUM instances on $n$ integers from $[\pm n^2]$ requires $n^{2-o(1)}$ time. Previously, Abboud, Bringmann, Khoury and Zamir \cite{AbboudBKZ22} showed an $m^{4/3-o(1)}$ lower bound for Sparse Triangle Detection  under an unbalanced bounded-weight variant of the Exact Triangle hypothesis.

The method for our second reduction can actually lead to more consequences, including an %
$m^{5/4-o(1)}$
lower bound of the All-Nodes Sparse Triangle problem (deciding, for each node, whether it is in a triangle) under the Exact Triangle hypothesis, and an $m^{9/7-o(1)}$ lower bound of the Sparse Triangle Detection problem under a (more natural) balanced bounded-weight variant of the Exact Triangle hypothesis.

\section{Exact Triangle with Few Zero-Weight \texorpdfstring{$4$}{4}-Cycles}
In this section, we will consider Exact Triangle instances where the edge weights $w_{ij}$ can be different from $w_{ji}$, and the weight of a triangle $(i, j, k)$ is defined as $w_{ij} + w_{jk} + w_{ki}$. The following property was considered by \cite{JinXstoc23}.

\begin{property}
\label{prop:fewC4}
In an Exact Triangle instance on edge set $E$ and weight function $w$, 
\begin{itemize}
    \item \textbf{Antisymmetry}: For every $(i, j) \in E$, it holds that $(j, i) \in E$ and $w_{ij} = -w_{ji}$;
    \item \textbf{Few zero-weight 4-cycles}: The number of $(i, j, k, \ell)$ for which $(i, j), (j, k), (k, \ell), (\ell, i) \in E$ and $w_{ij} + w_{jk} + w_{kl} + w_{li} = 0$ is at most $n^3$ (we will call each $(i, j, k, \ell)$ a zero-weight 4-cycle, and here, $i, j, k ,\ell$ are not necessarily distinct). 
\end{itemize}

\end{property}

Jin and Xu~\cite{JinXstoc23} gave a long and complicated proof that
Exact Triangle instances on graphs with \cref{prop:fewC4} require $n^{3-o(1)}$ time under the 3SUM hypothesis.  We present a much simpler proof of the same result under the Exact Triangle hypothesis (from which one can then obtain conditional lower bounds for a host of other problems, including 4-Cycle-Enumeration and Offline Approximate Distance Oracle, as shown in \cite{JinXstoc23}):

\begin{theorem}
Under the Exact Triangle hypothesis, Exact Triangle instances on $n$-node graphs with \cref{prop:fewC4} requires $n^{3-o(1)}$ time. 
\end{theorem}
\begin{proof}
Suppose for the sake of contradiction that Exact Triangle instances on $n$-node graphs with \cref{prop:fewC4} can be solved in $O(n^{3-\eps})$ time for some $\eps > 0$. 

Given any Exact Triangle instance on $G = (V, E)$ (where $w_{ij} = w_{ji}$ for all edges $ij$), we can copy $V$ three times to $V_1, V_2, V_3$, add $E$ between any two copies of $V$, and direct the edges in the direction of $V_1 \rightarrow V_2 \rightarrow V_3 \rightarrow V_1$. This way, we can add in the edges in the reverse directions with negated weights. Hence, the Antisymmetry property is satisfied. 

Let $\delta > 0$ be some constant to be fixed later. We consider the following two cases:
\paragraph{The number of zero-weight $4$-cycles is at most $n^{4-\delta}$.} In this case, we randomly partition the vertices to reduce the problem to multiple smaller Exact Triangle instances, reducing the overall number of zero-weight $4$-cycles significantly while preserving all zero-weight triangles. 

Let $x > 0$ be a constant. We create $n^{1-x}$ buckets of vertices, and put each $v \in V$ into one of the buckets chosen uniformly at random. For every triple of buckets, we create an Exact Triangle instance on the subgraph induced by the vertices in the three buckets. Each instance contains $O(n^x)$ nodes w.h.p.\footnote{With high probability, i.e., probability $1-O(1/n^c)$ for an arbitrarily large constant $c$.} It suffices to solve all $O(n^{3-3x})$ instances of these small Exact Triangle instances. 

For any triple of buckets $B_1, B_2, B_3$, and any zero-weight $4$-cycle $(i, j, k, \ell)$ in the original graph with distinct $i, j, k, \ell$, $(i, j, k, \ell)$ appears in the small instance induced by $B_1, B_2, B_3$ with probability $\frac{1}{n^{4-4x}}$. Therefore, in expectation, the number of zero-weight $4$-cycles with distinct nodes in the instance induced by buckets $B_1, B_2, B_3$ is $O(n^{4x-\delta})$. By  sampling random distinct tuples of nodes $i, j, k, \ell$ in the small instance and testing whether they form a zero-weight $4$-cycle, we can estimate the number of zero-weight $4$-cycles with distinct vertices up to $n^{2x}$ additive error w.h.p.\ in $\tO(\frac{n^{4x}}{n^{2x}}) = \tO(n^{2x})$ time. If the estimated count is at least $2n^{2x}$, we use brute-force to solve this small instance in $O(n^{3x})$ time. Note that this case happens with probability $O(\frac{n^{4x-\delta}}{n^{2x}}) = O(n^{2x - \delta})$ by Markov's inequality. 
Otherwise, the number of zero-weight $4$-cycles in the instance is at most $O(n^{3x})$ w.h.p.\  ($O(n^{2x})$ from zero-weight $4$-cycles with distinct nodes, and $O(n^{3x})$ from zero-weight $4$-cycles with some repeated nodes). We can add some isolated nodes in the small instance to increase its number $n'$ of nodes  ($n'$ will still be $O(n^x)$), so that the number of zero-weight $4$-cycles can be bounded by $(n')^3$. Thus, we can call an algorithm for Exact Triangle with \cref{prop:fewC4} on this small instance, which runs in $O(n^{(3-\eps)x})$ time. 

Overall, the expected running time is 
$$O\left(n^{3-3x} \cdot \left(n^{2x-\delta} \cdot n^{3x} + n^{(3-\eps)x}\right) \right) = O(n^{3+2x-\delta} + n^{3-3\eps x}),$$
which is $O(n^{3-\frac{3\eps}{2+3\eps}\delta})$ by setting $x = \delta / (2+3\eps)$. 

\paragraph{The number of zero-weight $4$-cycles is more than $\Omega(n^{4-\delta})$.} In this case, 
there exists $(i_0, k_0) \in E$ that participates in $\Omega(n^{2-\delta})$ zero-weight $4$-cycles.
We can find such an $(i_0,k_0)$ by testing all $O(n^2)$ pairs.
To test a given pair $(i_0,k_0)$, we can estimate the number of $(i,k)$ such that $(i_0,k_0,i,k)$ form
a zero-weight $4$-cycle
with up to $o(n^{2-\delta})$ additive error w.h.p., by sampling random pairs $(i,k)$,\ in 
$\OO(\frac{n^2}{n^{2-\delta}})=\OO(n^\delta)$ time.  The total time is $\OO(n^{2+\delta})$.

For each $k$, define $r_k := w_{i_0, k} - w_{i_0, k_0}$. Define $E_0 := \{(i, k) \in E: w_{ik} - w_{i,k_0}=r_k\}$. Note that $E_0$ is exactly the set of $(i, k)$ where $(i_0, k_0, i, k)$ forms a  zero-weight $4$-cycle, because
$$w_{ik}-w_{i,k_0} = w_{i_0, k} - w_{i_0, k_0} \iff w_{i_0, k_0} + w_{k_0, i} + w_{ik} + w_{k, i_0} = 0$$
by the Antisymmetry property (and Fredman's trick). Thus, 
$|E_0| = \Omega(n^{2-\delta})$ w.h.p. 

We can test whether there is a zero-weight triangle using at least one edge in $E_0$ %
as follows. Say $w_{ik} + w_{kj} + w_{ji} = 0$ for some edge $(i, k) \in E_0$. The condition is equivalent to $w_{i,k_0}+r_k+w_{kj} + w_{ji} = 0$, which further is equivalent to $w_{i,k_0}+w_{ji} = -w_{kj}-r_k$. We create two matrices $A, B$ where  $A_{ij} := w_{i, k_0} + w_{ji}$ and $B_{jk} := -w_{kj} - r_k$. Now the problem becomes the following: for every $(i, k) \in E_0$, test whether there exists $j$ such that $A_{ij} = B_{jk}$. This can be solved by computing the so-called \emph{Equality Product} of the two matrices $A$ and $B$, which is known to be in $O(n^{(3+\omega)/2})$ time \cite{MatIPL} (see \cite{CVXstoc23} for more applications of Equality Product combined with Fredman's trick). After this step, we can remove $E_0$ from $E$, and by symmetry, we can also remove all edges that are reverses to the edges in $E_0$.  Then we can iteratively solve the Exact Triangle instance on the remaining graph. 

\paragraph{Putting things together. } The overall algorithm is as follows. It determines which case we are at by estimating the number of zero-weight $4$-cycles in $\tO(n^{\delta})$ time. We can call the second case at most $O(n^{\delta})$ expected number of times because each time we remove a subset of $\Omega(n^{2-\delta})$ edges w.h.p.
Also, we can call the first case at most once. Thus, the overall running time is 
$$\OO(n^{2\delta} + n^{3-\frac{3\eps}{2+3\eps}\delta} + n^\delta \cdot (n^{2+\delta}+n^{(3+\omega)/2})),$$
which is truly subcubic by setting $\delta$ properly (say $\delta = 0.1$). This would violate the Exact Triangle hypothesis, so Exact Triangle instances on $n$-node graphs with \cref{prop:fewC4} requires $n^{3-o(1)}$ time. 
\end{proof}

\section{All-Edges Sparse Triangle: First Alternative Proof}
\label{sec:aesparsetri1}
In this section, we present the first alternative proof for the lower bound of All-Edges Sparse Triangle under the  Exact Triangle hypothesis (from which one can then obtain conditional lower bounds for many data structure problems \cite{patrascu2010towards,abboud2014popular,KopelowitzPP16}). This proof is adapted from Chan, Vassilevska Williams and Xu's reduction from real-valued problems to All-Edges Sparse Triangle \cite{CVXstoc22}.  We combine it with an extra simple idea where we iteratively halve the weights.  The new reduction is entirely deterministic, unlike all previous reductions for All-Edges Sparse Triangle, and avoids any form of hashing. 

\begin{theorem}
Under the Exact Triangle hypothesis, All-Edges Sparse Triangle requires $m^{4/3-o(1)}$ time. 
\end{theorem}
\begin{proof}
Without loss of generality, we assume that the Exact-Triangle instance is on a tripartite graph with node partitions $A, B, C$ where $|A| = |B| = n$ and $|C| = \sqrt{n}$ and with weight function $w$. Such an instance requires $n^{2.5-o(1)}$ time under the Exact-Triangle hypothesis. 
We also solve a stronger version, where for every $(a, b) \in A \times B$, we need to find a $c_{ab} \in C$ where $|w_{ab}+w_{b,c_{ab}}+w_{a,c_{ab}}| \le 3$, if one exists. This version is stronger because if we multiply all initial edge weights by $4$, then finding such $(a, b, c_{ab})$ is equivalent to finding zero-weight triangles. 

First, we recursively solve the instance where the edge weights $w_{ij}$ are replaced with $\lfloor w_{ij} / 2 \rfloor$. The base case is when all the edge weights are in $[\pm O(1)]$. In this case, the instance can be solved in $\tO(n^\omega)$ time, by enumerating all triples of edge weights that sum up to $[\pm 3]$, and then find triangles with the corresponding edge weights using known methods for finding triangles in unweighted graphs \cite{AlonGMN92}. 

Suppose for every $(a, b)$, we already found $k_{ab}$ (via recursion) where $\left|\lfloor w_{ab} / 2 \rfloor + \lfloor w_{b,k_{ab}} / 2 \rfloor + \lfloor w_{a,k_{ab}} / 2 \rfloor \right| \le 3$ if one exists. Then we need to find $c_{ab}$ with $\left|w_{ab}+w_{b,c_{ab}}+w_{a,c_{ab}}\right| \le 3$ if one exists. 

Let $L_{k,\delta} := \{(a, b) \in A \times B: k_{ab} = k, w_{ab}+w_{bk}+w_{ka} = \delta\}$ for $k \in C$ and $\delta \in \{-6, \ldots, 9\}$. If some $L_{k,\delta}$ has size greater than $n^{1.5}$, we split it to $O(\frac{|L_{k,\delta}|}{n^{1.5}})$ subsets of size $O(n^{1.5})$. The total number of subsets will remain $O(|C|+\frac{n^2}{n^{1.5}})=O(n^{0.5})$. We also split the set $C$ into $n^\eps$ equally-sized subsets $\left\{C_s\right\}_{s \in [n^\eps]}$. For each subset $L = L_{k, \delta}$, each $\delta' \in \{-3, \ldots, 3\}$, and each subset $C_s$, we create the following All-Edges Sparse Triangle instance:

\begin{itemize}
    \item The vertex sets are $A', B', U$, where $A'$ is a copy of $A$, $B'$ is a copy of $B$, and $U$ is $C_s \times [\pm n^{O(1)}]$ (sparsely represented). 
    \item For every $(a, b) \in L$, add an edge between $a \in A'$ and $b \in B'$. 
    \item For every $(a, c) \in A \times C_s$, add an edge between $a$ and $(c, w_{ac}-w_{ak}+\delta - \delta')$. 
    \item For every $(b, c) \in B \times C_s$, add an edge between $b$ and $(c, w_{bk}-w_{bc})$. 
\end{itemize}

Then we detect for every edge $(a, b)$ in the graph, whether it is contained in a triangle. If so, we exhaustively search $c_{ab}$ in $C_s$ (for each edge $(a, b)$, we only perform this exhaustive search once). 

\paragraph{Correctness.} 

First, we show that for any triangle $(a, b, (c, u))$ that is in any of the created All-Edges Sparse Triangle instances, $c$ is a valid candidate for $c_{ab}$. This is because $w_{ac}-w_{ak}+\delta - \delta' = u = w_{bk}-w_{bc}$, which implies $w_{ab}+w_{bc}+w_{ca} = \delta'+(w_{ab}+w_{bk}+w_{ak} - \delta)$ (by Fredman's trick). 
Now, $w_{ab}+w_{bk}+w_{ak} - \delta = 0$ because any edge $(a, b)$ in this particular All-Edges Sparse Triangle instance has $w_{ab}+w_{bk}+w_{ak} = \delta$ by the definition of $L_{k, \delta}$. Therefore, $w_{ab}+w_{bc}+w_{ca} = \delta' \in [\pm 3]$ and thus $c$ is a valid candidate for $c_{ab}$.

Then we show that if $c \in C_s$ is a valid candidate for $c_{ab}$, then $(a, b, (c, u))$ for some $u$ is a triangle in some of the created All-Edges Sparse Triangle instances. 
Because $c$ is a valid candidate for $c_{ab}$, and 
$$\frac{1}{2} \left(w_{ab}+w_{bc}+w_{ca} \right) - \frac{3}{2} \le \lfloor w_{ab} / 2 \rfloor + \lfloor w_{bc} / 2 \rfloor + \lfloor w_{ca} / 2 \rfloor \le \frac{1}{2} \left(w_{ab}+w_{bc}+w_{ca} \right).$$
we have $-3 \le \lfloor w_{ab} / 2 \rfloor + \lfloor w_{bc} / 2 \rfloor + \lfloor w_{ac} / 2 \rfloor \le \frac{3}{2} \le 3$, so $k_{ab}$ must also exist (in particular, it can take value $c$, but it can be any other valid value as well). Therefore, $(a, b) \in L_{k_{ab}, \delta}$, where $\delta = w_{ab}+w_{a,k_{ab}}+w_{b,k_{ab}} \in \{-6, \ldots, 9\}$. Also, let $\delta' = w_{ab}+w_{ac}+w_{bc}$. 
The edge $(a, b)$ will be included in the All-Edges Sparse Triangle instance created for  $L = L_{k_{ab}, \delta}$, $\delta'$ and $C_s$. Because $w_{ab}+w_{ak}+w_{bk} - \delta = 0 = w_{ab}+w_{ac}+w_{bc} - \delta'$, we have $w_{ac} - w_{ak} + \delta - \delta' = w_{bk}-w_{bc}$, so we will find a triangle for the edge $(a, b)$ (in particular, $(a, b, (c, w_{bk}-w_{bc}))$ is a valid triangle).

\paragraph{Running time.} 
Suppose the running time for All-Edges Sparse Triangle is $O(m^{4/3-\eps})$ for some $\eps > 0$, then the overall running time can be bounded as follows. The time for handling the base case is $\tO(n^\omega)$. In each recursive step, we need to solve $O(n^{0.5 +\eps})$ instances of All-Edges Sparse Triangle, which run in  $\tO(n^{0.5+\eps} \cdot (n^{1.5})^{4/3-\eps})=\tO(n^{2.5-0.5\eps})$ time. For every edge $(a, b)$, we also need to spend time performing the exhaustive search, which run in $\tO(n^{2} \cdot n^{0.5-\eps}) = O(n^{2.5-\eps})$ time. Overall, the running time is $\tO(n^\omega+n^{2.5 - 0.5\eps})$, which would violate the Exact Triangle hypothesis. 
\end{proof}

\section{All-Edges Sparse Triangle: Second Alternative Proof}
\label{sec:aesparsetri2}

In this section, we  present the second alternative proof for the lower bound of All-Edges Sparse Triangle under the  Exact-Triangle hypothesis, which is based on an idea from the recent conditional lower bound of Sparse Triangle Detection under the strong $3$SUM hypothesis by Jin and Xu \cite{JinXstoc23}. 

Let $\ExactTri(n\mid W)$ denote the time complexity of the problem of detecting a zero-weight triangle in an $n$-node tripartite graph with edge weights in $[\pm W]$.
Let $\ExactTriList(n,t\mid W)$ denote the time complexity of the problem of listing all $t$
zero-weight triangles in an $n$-node tripartite graph with edge weights in $[\pm W]$ (note that the value of $t$ is not given in advance).

Let $\SparseTri(m)$, $\AESparseTri(m)$, and $\ANSparseTri(m)$ denote the time complexity of the Sparse Triangle Detection, All-Edges Sparse Triangle, and All-Nodes Sparse Triangle problems respectively for an $m$-edge unweighted tripartite graph.\footnote{As a reminder, in the Sparse Triangle Detection problem, we need to determine whether a given graph contains a triangle; in the All-Nodes Sparse Triangle problem, we need to decide whether each node in a given graph is contained in a triangle. }
Let $\SparseTriList(m,t)$ denote the time complexity of the problem of listing all $t$ triangles in an $m$-edge unweighted tripartite graph.

Our main idea for reducing Exact Triangle to Sparse Triangle Listing is encapsulated in the proof of the following lemma, which is short and simple.

\begin{lemma}\label{lem:reduce1}
$\ExactTriList(n,t\mid W)\le O(1)\cdot \SparseTriList(n^2W^{1/3},t)$.
\end{lemma}
\begin{proof}
Exact Triangle Listing for weights in $[\pm W]$ reduces to $O(1)$
instances of Exact Triangle Listing for weights in $[\pm q]^3$ with $q:=W^{1/3}$, since
we can map any number $x=x_1q^2+x_2q+x_3\in [\pm W]$ with $x_1\in [\pm q]$, $x_2,x_3\in\{0,\ldots,q-1\}$ to the triple $\varphi(x)=(x_1,x_2,x_3)\in [\pm q]^3$; then $x+y+z=0$ iff $\varphi(x)+\varphi(y)+\varphi(z)\in\Delta$
for a constant-size set $\Delta=\{(0,0,0),(0,-1,q),(0,-2,2q)\}+\{(0,0,0),(-1,q,0),(-2,2q,0)\}$. For every $\delta \in \Delta$, we need to solve an instance where the target weight of the triangle is $\delta$; by subtracting $\delta$ from all edges in one of the edge parts, the target can become $0$ again (this will only increase the edge weight bound by a constant factor). 

Let $G$ be an $n$-node tripartite graph with node partition $A,B,C$ and edge weights in $[\pm q]^3$.   
Construct a new unweighted tripartite graph $G'$ with node partition $A\times [\pm q]^2$,
$B\times [\pm q]^2$, $C\times [\pm q]^2$:
\begin{itemize}
\item Add an edge between $(a,x_1,z_3)\in A\times [\pm q]^2$ and $(b,y_3,x_2)\in B\times [\pm q]^2$
iff $ab$ is an edge in $G$ with weight $(x_1,x_2,-y_3-z_3)$.
\item Add an edge between $(b,y_3,x_2)\in B\times [\pm q]^2$ and $(c,z_2,y_1)\in C\times [\pm q]^2$
iff $bc$ is an edge in $G$ with weight $(y_1,-x_2-z_2,y_3)$.
\item Add an edge between $(c,z_2,y_1)\in C\times [\pm q]^2$ and $(a,x_1,z_3)\in A\times [\pm q]^2$
iff $ca$ is an edge in $G$ with weight $(-x_1-y_1,z_2,z_3)$.
\end{itemize}

The number of edges in $G'$ is $O(n^2q)=O(n^2W^{1/3})$ (for example, to count the number of edges added in the first bullet, note that there are $O(q)$ ways to write a given number in $[\pm q]$ as $-y_3-z_3$).  
Furthermore, 
$(a,x_1,z_3),(b,y_3,x_2),(c,z_2,y_1)$ form a triangle in $G'$ iff $a,b,c$ form a triangle in $G$ whose edges $ab$, $bc$, $ca$ have
weights $(x_1,x_2,x_3)$, $(y_1,y_2,y_3)$, $(z_1,z_2,z_3)$ summing to $(0,0,0)$.
Thus, we can solve Exact Triangle Listing on $G$ by solving Triangle Listing on $G'$.
\end{proof}

For the above lemma to be effective, the universe size $W$ needs to be subcubic. 
To lower the universe size, we use a standard 
hashing trick described in the next lemma.  We only need the simplest type of linear hash functions (mod $p$ for a random prime $p$), and this is the only place in this reduction where randomization is used. 

\begin{lemma}\label{lem:reduce2}
$\ExactTri(n\mid n^{O(1)}) \le \OO(1)\cdot \ExactTriList(n,\OO(t)\mid n^3/t)$ for any $t$.
\end{lemma}
\begin{proof}
Pick a random prime $p$ in $[n^3/(2t), n^3/t]$.
To solve Exact Triangle for a given weighted graph,
we take the weights mod $p$ and solve Exact Triangle Listing to find all
triangles with weights congruent to $0 \bmod p$.  The answer is yes iff
one of these triangles actually has zero weight.
For a fixed triangle with nonzero weight, the probability that its weight is congruent to
$0\bmod p$ is $O((t/n^3)\log n)$,
since there are $O(\frac{n^3/t}{\log (n^3/t)})$ primes in $[n^3/(2t), n^3/t]$, and 
a fixed number in $[n^{O(1)}]$ has at most $O(\frac{\log n}{\log(n^3/t)})$ prime divisors in this range.  Thus, the expected number of triangles with nonzero weights congruent to $0\bmod p$ (i.e., number of false positives) is $O(t\log n)$.  Thus, if the process takes
more than $\ExactTriList(n,\OO(t)\mid n^3/t)$ time, we can terminate and infer that the answer is yes with probability $\Omega(1)$ (or w.h.p.\ by repeating logarithmically many times).
\end{proof}

The main result now immediately follows:

\begin{theorem}\label{thm:tri:list}
Under the Exact Triangle hypothesis, All-Edges Sparse Triangle requires $m^{4/3-o(1)}$ time. 
\end{theorem}
\begin{proof}
We first show that for listing all $t$ triangles in an $m$-edge graph when $t=\Theta(m)$ requires $m^{4/3-o(1)}$ time, under the Exact Triangle hypothesis.

Combining Lemmas~\ref{lem:reduce1} and \ref{lem:reduce2} gives
\begin{eqnarray*}
\ExactTri(n\mid n^{O(1)}) &\le& \OO(1)\cdot \SparseTriList(n^3/t^{1/3},\OO(t))\\
&\le& \OO(1)\cdot \SparseTriList(n^{9/4},\OO(n^{9/4}))) \qquad\mbox{by setting $t=n^{9/4}$}\\
&\le& \OO(n^{3-(9/4)\eps}) \qquad\mbox{if $\SparseTriList(m,m)\le O(m^{4/3-\eps})$}.
\end{eqnarray*}

To prove hardness for All-Edges Sparse Triangle instead of Sparse Triangle Listing, we apply a known (simple) reduction from Sparse Triangle Listing to All-Edges Sparse Triangle from \cite[Theorem 1.4]{DurajK0W20}. 
\end{proof}

\subsection{Discussion and Further Consequences}

Monte Carlo randomization is used in the proof of Lemma~\ref{lem:reduce2}.  With a little more effort, it is possible to modify it to require only Las Vegas randomization.%
\footnote{For example, in the proof of Lemma~\ref{lem:reduce2},
when the process takes more than the allotted time, we return ``not sure'' instead of yes.  If there is at most one zero-weight triangle, the probability of returning ``not sure'' is small.  By standard random-sampling arguments, we can run the algorithm on logarithmically many subgraphs where one of the subgraphs has a unique zero-weight triangle with good probability if there exists a zero-weight triangle.  If ``not sure'' is reported for all such subgraphs, we restart from scratch.}  However, it is not clear how to completely derandomize this reduction without slow-down.
If a deterministic reduction is desired, see our alternative proof in 
\cref{sec:aesparsetri1}.

Compared to Vassilevska Williams and Xu's reduction~\cite{williamsxumono} from Exact Triangle to Sparse Triangle Listing or All-Edges Sparse Triangle, the new reduction is not only simpler, but also
better in a technical sense.  Their reduction \cite[Theorem~3.4]{williamsxumono} (also randomized) basically shows that 
\[ \ExactTri(n\mid n^{O(1)})\le \OO(n^{2\rho})\cdot \SparseTriList(n^{2-\rho},n^{3-3\rho}).\]
Setting $\rho=1/2$ gives 
$\ExactTri(n\mid n^{O(1)})\le \OO(n)\cdot \SparseTriList(n^{3/2},n^{3/2})$,
which is $\OO(n^{3-(3/2)\eps})$ if $\SparseTriList(m,m)\le O(m^{4/3-\eps})$.
In contrast, the proof of Theorem~\ref{thm:tri:list} provides a slightly better bound of 
$\OO(n^{3-(9/4)\eps})$.  

As an application, we derive a new conditional lower bound for All-Nodes Sparse Triangle, a natural problem ``in between'' All-Edges Sparse Triangle and Sparse Triangle Detection.
(The previous reduction could also be used in combination of Lemma~\ref{lem:AN} below, but the dependence on $\eps$ would again be worse.)

\newcommand{\dmax}{d_{\mbox{\scriptsize\rm max}}}

Let $\SparseTriList(m,t,\dmax)$ denote the time complexity of the problem of listing all $t$ triangles in an $m$-edge unweighted tripartite graph with maximum degree $\dmax$.

\begin{lemma}\label{lem:reduce1:deg}
$\ExactTriList(n,t\mid W)\le \OO(1)\cdot \SparseTriList(n^2W^{1/3},t,O(n/W^{1/3}))$.
\end{lemma}
\begin{proof}
As in the proof of Lemma~\ref{lem:reduce1}, we may assume the edge weights are in $[\pm q]^3$ with $q:=W^{1/3}$. 
Pick a random function $h:A\cup B\cup C\rightarrow [q]^3$.
For each edge $uv$, replace the weight $w_{uv}$ by $\hat{w}_{uv} = w_{uv} + h(v) - h(u)\in [\pm 2q]^3$.
Clearly, the weight of each triangle remains unchanged.
Now, apply the same construction of $G'$ as in the proof of Lemma~\ref{lem:reduce1} (with $q$ adjusted by a factor of $2$).

Consider a fixed node $(a,x_1,z_3)\in A\times [\pm 2q]^2$, and a fixed $b\in B$.
Then $(a,x_1,z_3)$ has a (unique) neighbor $(b,y_3,x_2)$ in $G'$ for some $y_3$ and $x_2$
only if the first component of $\hat{w}_{ab}$ is equal to $x_1$, i.e.,
$h(b)$ is equal to $h(a)-w_{ab} + x_1$ in the first component.  
This holds with probability at most $1/q$, since the first component of $h(b)$
is uniformly distributed in $[q]$ conditioned to any fixed $h(a)$.
For a fixed $c\in C$,
the probability that $(a,x_1,z_3)$ has a (unique) neighbor $(c,z_2,y_1)$ for some $z_2$ and $y_1$ is similarly at most $1/q$.
Thus, the expected degree of $(a,x_1,z_3)$ is at most $n/q=n/W^{1/3}$.

We now remove all nodes of degree more than $6n/W^{1/3}$ from $G'$.
The probability that a fixed node is removed is at most $1/6$ by Markov's inequality.
Thus, the probability that any fixed triangle is eliminated is at most $1/2$.
To solve Exact Triangle Listing on $G$, we solve Sparse Triangle Listing on $G'$,
repeating logarithmically many times to ensure a high survival probability bound per triangle.
\end{proof}

\begin{lemma}\label{lem:AN}
$\SparseTriList(m,t,\dmax)\le \OO(\ANSparseTri(m+t\dmax))$. 
\end{lemma}
\begin{proof}
We describe a recursive algorithm for Sparse Triangle Listing:
Given an $m$-edge tripartite graph $G$ with node partition $A,B,C$, arbitrarily divide $A$ into two subsets $A_1$ and $A_2$ of size $|A|/2$.
For each $j\in [2]$, run the All-Nodes Sparse Triangle oracle to identify the subset $B_j$ of nodes in $B$ that participate in triangles in the subgraph of $G$ induced by $A_j\cup B\cup C$; then recursively solve the Sparse Triangle Listing problem on the subgraph of $G$ induced by $A_j\cup B_j\cup C$.  

Observe that each edge in $A\times B$ or $A\times C$ participates in one subproblem per level in the recursion.
Furthermore, an edge $(b,c)\in B\times C$ participates in at most $f(b)$ subproblems per level, where $f(b)$ denotes the number of triangles that $b$ participates in.  Thus, the sum of the number of edges over all the subproblems per level 
is $O(m + \sum_{b\in B} f(b)\dmax) = O(m+t\dmax)$.  There are $O(\log n)$ levels of recursion.
By super-linearity of $\ANSparseTri(\cdot)$, the entire algorithm runs in $O(\ANSparseTri(m+t\dmax)\log n)$ time.
\end{proof}

\begin{theorem}
Under the Exact Triangle hypothesis, All-Nodes Sparse Triangle requires $m^{5/4-o(1)}$ time.
\end{theorem}
\begin{proof}
Combining Lemmas \ref{lem:reduce2}, \ref{lem:reduce1:deg}  and \ref{lem:AN} gives 
\begin{eqnarray*}
\ExactTri(n\mid n^{O(1)}) &\le& \OO(1)\cdot  \ANSparseTri(n^3/t^{1/3} + t^{4/3})\\
&\le& \OO(1)\cdot \ANSparseTri(n^{12/5})  \qquad\mbox{by setting $t=n^{9/5}$}\\
&\le& \OO(n^{3-(12/5)\eps}) \qquad\mbox{if $\ANSparseTri(m)\le O(m^{5/4-\eps})$}.
\end{eqnarray*}

\vspace{-\bigskipamount}
\end{proof}

Unfortunately, a nontrivial lower bound for Sparse Triangle Detection is still out of reach  under the Exact Triangle hypothesis.  However, we have the following result if we assume the \emph{Strong Exact Triangle hypothesis}, namely, that Exact Triangle for integer weights in $[\pm n]$ requires
$n^{3-o(1)}$ time.  (Recent work has considered similarly defined Strong APSP hypothesis, Strong 3SUM hypothesis, and so on~\cite{CVXstoc23,JinXstoc23}.)   The proof here only needs Lemma~\ref{lem:reduce1}, so the reduction is even simpler (and deterministic).

\begin{theorem}\label{thm:tri:detect}
Under the Strong Exact Triangle hypothesis, Sparse Triangle Detection requires
$m^{9/7-o(1)}$ time.
\end{theorem}
\begin{proof}
The analog of Lemma \ref{lem:reduce1} for detection instead of listing
implies that $\ExactTri(n\mid W) \le O(1)\cdot \SparseTri(n^2W^{1/3})$, and so
$\ExactTri(n\mid n) \le O(\SparseTri(n^{7/3}))$.
\end{proof}

Abboud et al.~\cite{AbboudBKZ22} were the first to prove conditional lower bounds for Sparse Triangle
Detection under a variant of the Exact Triangle hypothesis (which they call ``the Strong Zero-Triangle Conjecture'').
More precisely, they showed a better lower bound of $m^{4/3-o(1)}$ under an
``unbalanced'' version of the hypothesis, specifically, that
Exact Triangle for a tripartite graph with $n$ nodes in the first part, and
$\sqrt{n}$ nodes in the second and third parts, and weights in $[\pm\sqrt{n}]$ requires
$n^{2-o(1)}$ time.  In contrast, Theorem~\ref{thm:tri:detect} assumes
the Strong Exact Triangle hypothesis in the
balanced case, which appears more natural.

Compared to Abboud et al.'s reduction (which is extremely simple), ours is more powerful in some sense.
Let $\ExactTri(n_1,n_2,n_3\mid W)$ denote the time complexity of Exact Triangle
for a tripartite graph with $n_1,n_2,n_3$ nodes in its three parts and edge weights
in $[\pm W]$.  They basically observed that 
\[ \ExactTri(n_1,n_2,n_3\mid W)\le
O(1)\cdot \SparseTri(n_1n_2+n_1n_3 + n_2n_3W).
\]
Setting $n_2=n_3=W=\sqrt{n}$ gives $\ExactTri(n,\sqrt{n},\sqrt{n} \mid \sqrt{n})\le O(\SparseTri(n^{3/2}))$,
which is $O(n^{2-(3/2)\eps})$ if $\SparseTri(m)\le O(m^{4/3-\eps})$.  In contrast,
straightforward modification of the proof of Lemma~\ref{lem:reduce1} yields a more general bound:
\begin{eqnarray*}
 \ExactTri(n_1,n_2,n_3\mid W) &\le& \min_{q_1,q_2,q_3:\ q_1q_2q_3=W} O(1)\cdot\SparseTri(n_1n_2q_3+n_1n_3q_2+n_2n_3q_1)\\
&=& O(1)\cdot\SparseTri((n_1n_2n_3)^{2/3}W^{1/3} + n_1n_2 + n_1n_3 + n_2n_3).
\end{eqnarray*}

\section{Acknowledgements}

We would like to thank Ce Jin for helpful discussions. 

\bibliographystyle{alphaurl} 
\bibliography{main}

\appendix
\section{A Simple Reduction from 3SUM to Sparse Triangle Listing}\label{app:3sum}

By combining with known  reductions from $3$SUM to Convolution-$3$SUM~\cite{patrascu2010towards,KopelowitzPP16,ChanH20} and from Convolution-3SUM to Exact Triangle~\cite{VWfindingcountingj}, the result in \cref{sec:aesparsetri2}
implies a reduction from $3$SUM to Sparse Triangle Listing.
In this appendix, we describe a more direct modification to obtain a simple
reduction from $3$SUM to Sparse Triangle Listing, without going through Convolution-$3$SUM and Exact Triangle.  
The proof is very similar to the one in \cref{sec:aesparsetri2} (based on an idea in  \cite{JinXstoc23}), and is also similar to ideas contained in a proof by Abboud, Bringmann and Fischer~\cite{AbboudBF23}, but we feel it is worthwhile to write out the proof explicitly, since it is simpler (and more ``symmetric'') than P{\u{a}}tra{\c{s}}cu's original reduction~\cite{patrascu2010towards} or Kopelowitz, Pettie and Porat's later reduction~\cite{KopelowitzPP16}, and may have pedagogical value (being the easiest to teach).  

Let $\ThreeSUM(n\mid W)$ denote the time complexity of the $3$SUM problem for $n$ numbers in $[\pm W]$.
Let $\ThreeSUMList(n,t\mid W)$ denote the time complexity of the problem of listing all $t$ triples of numbers summing to 0 for $n$ numbers in $[\pm W]$ (note that $t$ is not given in advance).

We first adapt Lemma~\ref{lem:reduce1}.

\begin{lemma}\label{lem:reduce1:$3$SUM}
$\ThreeSUMList(n,t\mid W)\le O(1)\cdot \SparseTriList(nW^{1/3},t)$.
\end{lemma}
\begin{proof}
$3$SUM for numbers in $[\pm W]$ reduces to $O(1)$
instances of $3$SUM for triples in $[\pm q]^3$ with $q:=W^{1/3}$, by the same mapping $\varphi$.

Let $A,B,C$ be three given sets of $n$ triples in $[\pm q]^3$. 
Construct an unweighted tripartite graph $G'$ with node partition $\{1\}\times [\pm q]^2$,
$\{2\}\times [\pm q]^2$, $\{3\}\times [\pm q]^2$:
\begin{itemize}
\item Add edge between $(1,x_1,z_3)\in \{1\}\times [\pm q]^2$ and $(2,y_3,x_2)\in \{2\}\times [\pm q]^2$
iff $(x_1,x_2,-y_3-z_3)\in A$.
\item Add edge between $(2,y_3,x_2)\in \{2\}\times [\pm q]^2$ and $(3,z_2,y_1)\in \{3\}\times [\pm q]^2$
iff $(y_1,-x_2-z_2,y_3)\in B$.
\item Add edge between $(3,z_2,y_1)\in \{3\}\times [\pm q]^2$ and $(1,x_1,z_3)\in \{1\}\times [\pm q]^2$
iff $(-x_1-y_1,z_2,z_3)\in C$.
\end{itemize}

The number of edges in $G'$ is $O(nq)=O(nW^{1/3})$.  
Observe that  
$(1,x_1,z_3),(2,y_3,x_2),(3,z_2,y_1)$ form a triangle in~$G'$ iff $(x_1,x_2,x_3)\in A$, $(y_1,y_2,y_3)\in B$, $(z_1,z_2,z_3)\in C$
sum to $(0,0,0)$ for some $x_3,y_2,z_1$.
Thus, we can solve $3$SUM Listing on $A,B,C$ by solving Triangle Listing on $G'$.
\end{proof}

\begin{lemma}\label{lem:reduce2:$3$SUM}
$\ThreeSUM(n\mid n^{O(1)}) \le \OO(1)\cdot \ThreeSUMList(n,\OO(t)\mid n^3/t)$ for any $t$.
\end{lemma}
\begin{proof}
Similar to the proof of Lemma~\ref{lem:reduce2}.
\end{proof}

\begin{theorem}\label{thm:tri:list:$3$SUM}
Under the $3$SUM hypothesis, %
All-Edges Sparse Triangle
requires $m^{4/3-o(1)}$ time.
\end{theorem}
\begin{proof}
Combining Lemmas~\ref{lem:reduce1:$3$SUM} and \ref{lem:reduce2:$3$SUM} gives
\begin{eqnarray*}
\ThreeSUM(n\mid n^{O(1)}) &\le& \OO(1)\cdot \SparseTriList(n^2/t^{1/3},\OO(t))\\
&\le& \OO(1)\cdot \SparseTriList(n^{3/2},\OO(n^{3/2}))) \qquad\mbox{by setting $t=n^{3/2}$}\\
&\le& \OO(n^{2-(3/2)\eps}) \qquad\mbox{if $\SparseTriList(m,m)\le O(m^{4/3-\eps})$}.
\end{eqnarray*}

To prove hardness for All-Edges Sparse Triangle instead of Sparse Triangle Listing, we can again apply a known reduction from Sparse Triangle Listing to All-Edges Sparse Triangle \cite{DurajK0W20}. 
\end{proof}

We remark that the proof of Lemma~\ref{lem:reduce1:$3$SUM} is similar to
a proof of Jin and Xu~\cite{JinXstoc23} that $\ThreeSUM(n\mid n^2) \le \OO(1)\cdot \SparseTri(n^{5/3})$,
implying an $m^{6/5-o(1)}$ lower bound for Sparse Triangle Detection under the Strong $3$SUM
hypothesis.  In fact, it is basically a reinterpretation.  The only main difference is that their proof used a different mapping $\varphi$
via the Chinese remainder theorem with three primes.  Their proof also used randomization, but the only reason was that they wanted the lower bound to hold for more structured graph instances with bounded
maximum degree.  A proof in Abboud, Bringmann and Fischer's work~\cite{AbboudBF23} also contained similar ideas (they also used randomized hashing to bound vertex degrees).

\begin{remark}
    Our first approach in \cref{sec:aesparsetri1} can also be adapted to give a direct reduction from $3$SUM to All-Edges Sparse Triangle, by combining the idea of halving the weights with the reduction from Real $3$SUM to \#All-Edges Sparse Triangle in \cite{CVXstoc23}.  The resulting reduction is deterministic and completely avoids hashing, unlike all previous reductions.  (Note that if we instead go through Convolution-3SUM and Exact Triangle, the part of the reduction from 3SUM to Convolution-3SUM would already need hashing~\cite{patrascu2010towards,KopelowitzPP16}, although this part has been derandomized~\cite{ChanH20}.)
\end{remark}

\end{document}